\documentclass[11pt]{article}

\usepackage{amsmath,amssymb,amsthm}
\usepackage{mathtools}
\usepackage{hyperref}
\usepackage{enumitem}
\usepackage{geometry}
\theoremstyle{definition}
\newtheorem{definition}{Definition}[section]

\theoremstyle{plain}
\newtheorem{theorem}[definition]{Theorem}
\newtheorem{proposition}[definition]{Proposition}

\theoremstyle{remark}
\newtheorem{remark}[definition]{Remark}

\newcommand{\Feas}{\mathrm{Feas}}
\newcommand{\WC}{\mathsf{WC}}

\title{Feasibility Preservation under Monotone Retrieval Truncation}
\author{Sean Plummer\\
Sr. AI Researcher, BambooHR}
\date{\today}

\begin{document}
\maketitle

\begin{abstract}
Retrieval-based systems approximate access to a corpus by exposing only a truncated subset of available evidence.
Even when relevant information exists in the corpus, truncation can prevent compatible evidence from co-occurring, leading to failures that are not captured by relevance-based evaluation.
This paper studies retrieval from a structural perspective, modeling query answering as a feasibility problem under truncation.

We formalize retrieval as a sequence of candidate evidence sets and characterize conditions under which feasibility in the limit implies feasibility at finite retrieval depth.
We show that monotone truncation suffices to guarantee finite witnessability for individual queries.
For classes of queries, we identify finite generation of witness certificates as the additional condition required to obtain a uniform retrieval bound, and we show that this condition is necessary.
We further exhibit sharp counterexamples demonstrating failure under non-monotone truncation, non-finitely-generated query classes, and purely slotwise coverage.

Together, these results isolate feasibility preservation as a correctness criterion for retrieval independent of relevance scoring or optimization, and clarify structural limitations inherent to truncation-based retrieval.
\end{abstract}

\section{Introduction}

Retrieval-based systems approximate access to large corpora by exposing only a limited subset of available evidence.
This truncation is unavoidable in practice, yet it introduces a structural source of failure: even when all information required to answer a query exists in the corpus, truncation may prevent compatible pieces of evidence from appearing together.
Such failures are commonly attributed to relevance or ranking errors, but they can arise even when retrieved items are individually appropriate.

This paper studies retrieval from a structural perspective.
Rather than asking whether retrieved evidence is relevant, we ask whether truncation preserves the existence of a feasible answer.
We treat query answering as a feasibility problem: a query is answerable if there exists a finite collection of evidence items that jointly satisfies the query’s constraints.
Retrieval is modeled as a truncation of the available evidence, and correctness is evaluated by whether feasibility is preserved under this truncation.

Within this framework, we identify simple structural conditions governing feasibility preservation.
First, we show that monotone truncation—where retrieved evidence is only added and never removed—suffices to guarantee finite witnessability for individual queries.
If a query is feasible in the limit, then under monotone retrieval a feasible witness must appear at some finite retrieval depth.
This property holds independently of relevance scoring, optimization objectives, or probabilistic assumptions.

For classes of queries, finite witnessability alone is insufficient.
We show that uniform retrieval bounds require an additional structural condition: finite generation of witness certificates.
When the set of certificates required across a query class is finite, a single retrieval depth suffices to validate all feasible queries in the class.
Conversely, we show that without finite generation, uniform bounds fail even under monotone retrieval.
This separation mirrors classical distinctions between compactness and Noetherianity, where existence of finite witnesses does not imply uniform bounds across families.

We further establish sharp boundaries for feasibility preservation.
We construct counterexamples showing that non-monotone truncation can destroy feasibility even when a feasible witness exists in the limit, that non-finitely-generated query classes preclude uniform retrieval budgets, and that local coverage of individual query slots does not guarantee global feasibility.
These examples demonstrate that feasibility is inherently relational and cannot, in general, be reduced to independent per-slot criteria.

The results presented here do not address the algorithmic complexity of solving general constraint satisfaction problems.
Instead, they isolate correctness conditions for truncation-based retrieval.
By separating feasibility preservation from relevance optimization, this work provides a structural account of retrieval failure modes, complementary to statistical and algorithmic analyses.

\section{Problem: Feasibility Preservation under Retrieval Truncation}

We study retrieval as a process that incrementally exposes a subset of available evidence, rather than as a mechanism for ranking or scoring documents.
The central question of this paper is not whether retrieved evidence is relevant, but whether truncation of the available evidence preserves the existence of a valid answer.

\paragraph{Evidence and queries.}
Let $U$ denote a universe of atomic evidence items.
A query $q$ specifies a finite set of informational requirements that must be jointly satisfied in order for the query to be answerable.
We model these requirements abstractly as constraints over $U$, without committing to any particular representation (e.g., documents, passages, or facts).

An answer to $q$ is witnessed by a finite collection of evidence items that together satisfy the query constraints.
We refer to such a collection as a \emph{feasible witness} for $q$.

\paragraph{Retrieval as truncation.}
A retrieval system does not expose the full universe $U$ at once.
Instead, it produces a sequence of increasingly large candidate sets
\[
D(1), D(2), \dots \subseteq U,
\]
where $D(k)$ represents the evidence made available at retrieval depth $k$. We do not assume this sequence is monotone unless stated explicitly. We interpret retrieval as a truncation of the available evidence, with $k$ serving as a resource or budget parameter.

We write
\[
D(\infty) := \bigcup_{k \ge 1} D(k)
\]
for the limiting set of evidence that would be available in the absence of truncation.

\paragraph{Feasibility under truncation.}
A query $q$ is said to be \emph{feasible at depth $k$} if there exists a feasible witness for $q$ that is contained entirely within $D(k)$.
Likewise, $q$ is feasible in the limit if it admits a feasible witness contained in $D(\infty)$.

In many retrieval-based systems, failures occur despite the existence of relevant information in the corpus.
From this perspective, such failures can be understood as violations of the following implication:
\[
\text{feasible in the limit} \;\;\Longrightarrow\;\; \text{feasible at some finite depth}.
\]

\paragraph{Problem statement.}
The goal of this paper is to characterize structural conditions on retrieval procedures under which feasibility is preserved under truncation.
More precisely, we ask:

\begin{quote}
\emph{Under what conditions on the sequence $\{D(k)\}$ does feasibility of a query in the limit guarantee feasibility at finite retrieval depth, and when can such a depth be chosen uniformly over classes of queries?}
\end{quote}

We emphasize that our focus is on \emph{existence of feasible witnesses}, not on ranking, scoring, or optimization of answers.
The results that follow identify simple structural properties of retrieval truncation that suffice to guarantee finite witnessability, as well as sharp boundaries where such guarantees provably fail.

\section{Model and Definitions}

This section formalizes the objects introduced informally in Section~2.
All subsequent results are stated and proved with respect to this model.

\subsection{Queries, Evidence, and Compatibility}

Let $U$ be a finite universe of atomic evidence items.
A query $q$ specifies a finite collection of informational requirements that must be satisfied jointly.

We model a query $q$ by:
\begin{itemize}[leftmargin=2em]
  \item a finite index set of slots $V(q) = \{1,\dots,m_q\}$,
  \item for each slot $i \in V(q)$, a set of admissible evidence items $A_i(q) \subseteq U$,
  \item a compatibility relation
  \[
  \mathcal{R}_q \subseteq A_1(q) \times \cdots \times A_{m_q}(q).
  \]
\end{itemize}

A \emph{witness} for query $q$ is a tuple
\[
(a_1,\dots,a_{m_q}) \in \mathcal{R}_q.
\]
Intuitively, a witness selects one admissible evidence item per slot in a way that satisfies all cross-slot constraints.

\begin{remark}
The compatibility relation $\mathcal{R}_q$ allows arbitrary finite-arity constraints and need not decompose into unary or pairwise conditions.
All results in this paper concern existence of witnesses and do not assume any particular logical form of $\mathcal{R}_q$ beyond finitarity.
\end{remark}

\subsection{Retrieval as Truncation}

A retrieval procedure induces a sequence of candidate evidence sets
\[
D(1), D(2), \dots \subseteq U,
\]
where $D(k)$ denotes the evidence made available at retrieval depth $k$.

We define the limiting candidate set
\[
D(\infty) := \bigcup_{k \ge 1} D(k).
\]

\begin{definition}[Slotwise Domains]
For a query $q$ and retrieval depth $k$, the induced domain for slot $i \in V(q)$ is
\[
D_i(q,k) := A_i(q) \cap D(k).
\]
\end{definition}

\begin{definition}[Feasibility]
A query $q$ is \emph{feasible at depth $k$}, denoted $\Feas(q,k)$, if there exists a witness
\[
(a_1,\dots,a_{m_q}) \in \mathcal{R}_q
\]
such that $a_i \in D_i(q,k)$ for all $i \in V(q)$.

Similarly, $q$ is \emph{feasible in the limit}, denoted $\Feas(q,\infty)$, if feasibility holds with respect to $D(\infty)$.
\end{definition}

\subsection{Witness Certificates}

To reason about feasibility abstractly, we associate each query with a finite certificate of evidence items.

\begin{definition}[Witness Certificate]
A \emph{witness certificate} for a query $q$ is a finite subset
\[
\WC(q) \subseteq U
\]
whose joint availability suffices to certify feasibility.
\end{definition}

The following soundness assumption links obligation sets to feasibility.

\begin{definition}[Certificate Soundness]
A witness certificate assignment $\WC(\cdot)$ is \emph{sound} if for all queries $q$ and depths $k$,
\[
\WC(q) \subseteq D(k) \;\Longrightarrow\; \Feas(q,k).
\]
\end{definition}

\begin{remark}
Witness certificates are \emph{sufficient} but not necessarily minimal or unique. Soundness and limit completeness together ensure that certificates act as finite feasibility witnesses without requiring necessity or minimality.
\end{remark}

\begin{definition}[Limit Completeness]
A witness certificate assignment $\WC(\cdot)$ is \emph{limit-complete} if for all queries $q$,
\[
\Feas(q,\infty) \;\Longrightarrow\; \WC(q) \subseteq D(\infty).
\]
\end{definition}

\subsection{Noetherian Retrieval}

We now formalize the notion of feasibility preservation under truncation.

\begin{definition}[Monotone Retrieval]
A retrieval procedure is \emph{monotone} if
\[
D(k) \subseteq D(k+1) \quad \text{for all } k \ge 1.
\]
\end{definition}

\begin{definition}[Noetherian Retrieval (NR)]
A retrieval procedure satisfies \emph{Noetherian Retrieval} if, for every query $q$,
\[
\Feas(q,\infty) \;\Longrightarrow\; \exists K \ge 1 \text{ such that } \Feas(q,K).
\]
\end{definition}

\begin{definition}[Uniform Noetherian Retrieval (UNR)]
Let $\mathcal{Q}$ be a class of queries.
A retrieval procedure satisfies \emph{Uniform Noetherian Retrieval} over $\mathcal{Q}$ if
\[
\Big( \forall q \in \mathcal{Q},\; \Feas(q,\infty) \Big)
\;\Longrightarrow\;
\exists K \ge 1 \text{ such that } \forall q \in \mathcal{Q},\; \Feas(q,K).
\]
\end{definition}

These definitions separate per-query finite witnessability from the stronger requirement of a uniform retrieval budget over a class of queries.

\section{Main Results}

\subsection{Monotone Retrieval and Finite Witnessability}

We begin with the simplest form of feasibility preservation, which concerns individual queries rather than classes of queries.

Intuitively, if a query admits a feasible witness in the limit and retrieval proceeds by monotonically expanding the available evidence, then all evidence items required by a witness must appear at finite retrieval depth. We assume retrieval procedures define $D(k)$ for all $k \ge 1$.

The following theorem formalizes this observation. 

\begin{theorem}[Monotone Retrieval Implies Noetherian Retrieval]
\label{thm:monotone-implies-nr}
If a retrieval procedure is monotone, then it satisfies Noetherian Retrieval.
That is, for every query $q$,
\[
\Feas(q,\infty) \;\Longrightarrow\; \exists K \ge 1 \text{ such that } \Feas(q,K).
\]
\end{theorem}

\begin{proof}
Assume $\Feas(q,\infty)$.
Then there exists a witness
\[
(a_1,\dots,a_{m_q}) \in \mathcal{R}_q
\]
such that $a_i \in D(\infty)$ for all $i \in V(q)$.

By definition of $D(\infty)$, for each $i$ there exists $k_i \ge 1$ such that $a_i \in D(k_i)$.
Let
\[
K := \max_{i \in V(q)} k_i.
\]
By monotonicity of retrieval, $D(k_i) \subseteq D(K)$ for all $i$, hence $a_i \in D(K)$ for every slot.

Therefore, $(a_1,\dots,a_{m_q})$ is a witness contained entirely in $D(K)$, and $\Feas(q,K)$ holds.
\end{proof}

\begin{remark}
Theorem~\ref{thm:monotone-implies-nr} does not rely on witness certificates or certificate soundness.
It is a purely structural consequence of monotone truncation and finitary witnesses.
\end{remark}

\subsection{Uniform Noetherian Retrieval}

Theorem~\ref{thm:monotone-implies-nr} guarantees finite witnessability on a per-query basis.
We now consider the stronger question of whether a \emph{single} retrieval depth can suffice to validate an entire class of queries.

Uniform bounds of this kind are not guaranteed in general and require additional structure.
We show that finite-generation conditions on witness certificates are sufficient to obtain such bounds.

\paragraph{Finite generation of witness certificates.}
Let $\mathcal{Q}$ be a class of queries.
Recall that each query $q \in \mathcal{Q}$ is associated with a witness certificate $\WC(q) \subseteq U$.

We define the \emph{basis} of the query class $\mathcal{Q}$ by
\[
B(\mathcal{Q}) := \bigcup_{q \in \mathcal{Q}} \WC(q).
\]

\begin{definition}[Finite Generation]
A query class $\mathcal{Q}$ is \emph{finitely generated} if there exists a finite set $G \subseteq U$ such that
\[
\WC(q) \subseteq G \quad \text{for all } q \in \mathcal{Q}.
\]
\end{definition}

Equivalently, $\mathcal{Q}$ is finitely generated if $B(\mathcal{Q})$ is finite.

\paragraph{Uniform feasibility under truncation.}
We now state the main uniformity result.

\begin{theorem}[Finite Generation Implies Uniform Noetherian Retrieval]
\label{thm:finite-gen-implies-unr}
Let $\mathcal{Q}$ be a class of queries.
Suppose that:
\begin{enumerate}[label=(\roman*)]
  \item the retrieval procedure is monotone,
 \item the witness certificate assignment $\WC(\cdot)$ is sound and limit-complete,
  \item $\mathcal{Q}$ is finitely generated, and
  \item for every $q \in \mathcal{Q}$, $\Feas(q,\infty)$ holds.
\end{enumerate}
Then there exists a retrieval depth $K \ge 1$ such that
\[
\Feas(q,K) \quad \text{for all } q \in \mathcal{Q}.
\]
That is, the retrieval procedure satisfies Uniform Noetherian Retrieval over $\mathcal{Q}$.
\end{theorem}

\begin{proof}
By finite generation, there exists a finite set $G \subseteq U$ such that $\WC(q) \subseteq G$ for all $q \in \mathcal{Q}$. By limit completeness and the assumption that $\Feas(q,\infty)$ holds for all $q \in \mathcal{Q}$, we have $\WC(q) \subseteq D(\infty)$ for all $q$, and hence $G \subseteq D(\infty)$.

By definition of $D(\infty)$, for each $g \in G$ there exists $k_g \ge 1$ such that $g \in D(k_g)$.
Let
\[
K := \max_{g \in G} k_g.
\]
By monotonicity of retrieval, $G \subseteq D(K)$.

Therefore, for every $q \in \mathcal{Q}$ we have $\WC(q) \subseteq D(K)$, and by certificate soundness $\Feas(q,K)$ holds.
\end{proof}

\begin{remark}
Theorem~\ref{thm:finite-gen-implies-unr} isolates finite generation of witness certificates as the structural condition required for uniform retrieval budgets.
Without this condition, uniform bounds may fail even under monotone retrieval.
\end{remark}

\section{Negative Results and Limitations}

\subsection{Failure under Non-Monotone Truncation}

We first show that monotonicity of retrieval is not merely a technical convenience, but a necessary condition for feasibility preservation.
In the absence of monotone truncation, feasibility in the limit need not imply feasibility at any finite depth, even for a single query.

\begin{proposition}[Non-Monotone Truncation Can Destroy Feasibility]
\label{prop:non-monotone-failure}
There exists a query $q$ and a retrieval procedure such that $\Feas(q,\infty)$ holds, but $\Feas(q,k)$ fails for all finite $k$.
\end{proposition}

\begin{proof}
Let the universe of evidence be
\[
U := \{a,b\}.
\]
Consider a query $q$ with two slots $V(q)=\{1,2\}$ and admissible sets
\[
A_1(q) = \{a\}, \qquad A_2(q) = \{b\}.
\]
Define the compatibility relation
\[
\mathcal{R}_q := \{(a,b)\}.
\]
Thus, $q$ has a unique feasible witness $(a,b)$.

Now define a retrieval procedure by the non-monotone sequence
\[
D(2k-1) := \{a\}, \qquad D(2k) := \{b\}
\]
for all $k \ge 1$.

Then
\[
D(\infty) = \{a,b\},
\]
and hence $\Feas(q,\infty)$ holds.
However, at every finite depth $k$, exactly one of $a$ or $b$ is present, so no feasible witness is contained in $D(k)$.
Therefore, $\Feas(q,k)$ fails for all finite $k$.
\end{proof}

\begin{remark}
Proposition~\ref{prop:non-monotone-failure} shows that even when a feasible witness exists in the limit, non-monotone truncation can prevent all of its components from co-occurring at any finite retrieval depth. This failure mode is eliminated by monotone retrieval, which ensures that once an evidence item appears it is never subsequently removed.
\end{remark}

\subsection{Failure of Uniform Bounds without Finite Generation}

We next show that finite generation of witness certificates is necessary for Uniform Noetherian Retrieval.
Even under monotone retrieval, the existence of feasible witnesses in the limit does not guarantee a uniform finite retrieval depth over a class of queries when certificates require infinitely many distinct evidence items.

\begin{proposition}[Failure of Uniform Bounds without Finite Generation]
\label{prop:non-finite-generation}
There exists a class of queries $\mathcal{Q}$ and a monotone retrieval procedure such that:
\begin{enumerate}[label=(\roman*)]
    \item $\Feas(q,\infty)$ holds for all $q \in \mathcal{Q}$,
    \item for every finite depth $K$, there exists $q \in \mathcal{Q}$ such that $\Feas(q,K)$ fails.
\end{enumerate}
Consequently, Uniform Noetherian Retrieval fails over $\mathcal{Q}$.
\end{proposition}

\begin{proof}
Let the universe of evidence be the countably infinite set
\[
U := \{e_1, e_2, e_3, \dots\}.
\]

Define a monotone retrieval procedure by
\[
D(k) := \{e_1, e_2, \dots, e_k\}.
\]

For each $i \ge 1$, define a query $q_i$ with a single slot $V(q_i)=\{1\}$, admissible set
\[
A_1(q_i) := \{e_i\},
\]
and compatibility relation
\[
\mathcal{R}_{q_i} := \{(e_i)\}.
\]
Thus, each query $q_i$ admits a unique feasible witness.

Define the witness certificate assignment by
\[
\WC(q_i) := \{e_i\}.
\]
This assignment is sound and limit-complete.

For every $i$, we have $e_i \in D(\infty)$, and hence $\Feas(q_i,\infty)$ holds.
However, for any fixed depth $K$, the query $q_{K+1}$ satisfies
\[
\WC(q_{K+1}) = \{e_{K+1}\} \not\subseteq D(K),
\]
and therefore $\Feas(q_{K+1},K)$ fails by certificate soundness.

Since the basis
\[
B(\mathcal{Q}) = \bigcup_{i \ge 1} \WC(q_i) = U
\]
is infinite, the query class $\mathcal{Q}$ is not finitely generated, and no uniform retrieval depth suffices.
\end{proof}

\begin{remark}
Proposition~\ref{prop:non-finite-generation} mirrors the classical failure of uniform bounds in non-Noetherian settings, where ascending chains do not stabilize.
Here, monotonicity alone guarantees per-query finite witnessability, but finite generation is required to obtain a uniform retrieval budget.
\end{remark}

\begin{remark}
The example can be truncated to arbitrarily large finite universes, showing that the failure persists in practical finite settings.
\end{remark}

\subsection{Slotwise Coverage Does Not Imply Feasibility}

Finally, we show that feasibility cannot in general be reduced to independent coverage of individual slots.
Even when each slot admits admissible evidence at a given retrieval depth, global feasibility may fail due to cross-slot compatibility constraints.

\begin{proposition}[Slotwise Coverage Is Insufficient]
\label{prop:slotwise-insufficient}
There exists a query $q$ and a retrieval depth $k$ such that:
\begin{enumerate}[label=(\roman*)]
    \item for every slot $i \in V(q)$, the induced domain $D_i(q,k)$ is non-empty, but
    \item $\Feas(q,k)$ fails.
\end{enumerate}
\end{proposition}

\begin{proof}
Let the universe of evidence be
\[
U := \{a_1, a_2, b_1, b_2\}.
\]
Consider a query $q$ with two slots $V(q)=\{1,2\}$ and admissible sets
\[
A_1(q) := \{a_1, a_2\}, \qquad A_2(q) := \{b_1, b_2\}.
\]

Define the compatibility relation
\[
\mathcal{R}_q := \{(a_1,b_1), (a_2,b_2)\}.
\]
Thus, admissible evidence items exist for each slot individually, but only matched pairs are jointly feasible.

Let the retrieval procedure satisfy
\[
D(k) := \{a_1, a_2, b_1\}.
\]
Then the induced domains satisfy
\[
D_1(q,k) = \{a_1, a_2\}, \qquad D_2(q,k) = \{b_1\},
\]
so both slots are locally covered.
However, no witness in $\mathcal{R}_q$ is contained entirely in $D(k)$, since $(a_2,b_2)$ requires $b_2 \notin D(k)$.
Therefore, $\Feas(q,k)$ fails.
\end{proof}

\begin{remark}
Proposition~\ref{prop:slotwise-insufficient} demonstrates that feasibility is inherently a global property.
Local coverage of individual slots does not guarantee the existence of a compatible witness.
This motivates modeling query constraints via joint compatibility relations and, correspondingly, the use of witness certificates rather than per-slot relevance criteria.
\end{remark}

\section{Complexity and Scalability}

We briefly address computational considerations and clarify how the feasibility framework studied in this paper relates to classical complexity results for constraint satisfaction problems.

\paragraph{Relation to CSP hardness.}
The feasibility problem induced by a query can be viewed as a finite constraint satisfaction problem (CSP).
In full generality, CSPs and, in particular, $k$-SAT are NP-hard.
However, the results of this paper do not concern the complexity of solving arbitrary CSPs.

Our guarantees are structural rather than algorithmic.
They characterize conditions under which feasibility is preserved under truncation, not worst-case bounds for deciding feasibility.

\paragraph{Bounded and parameterized setting.}
The operational regime considered here is highly restricted.
Each query induces:
\begin{itemize}[leftmargin=2em]
    \item a fixed and typically small number of slots,
    \item candidate domains bounded by retrieval depth,
    \item an existential feasibility check that terminates upon finding a single witness.
\end{itemize}

Under these conditions, feasibility checking is fixed-parameter tractable with respect to the number of slots and the sizes of the retrieved domains.
The dominant computational cost lies in retrieval, not in feasibility validation.

\paragraph{Separation of retrieval and validation costs.}
Retrieval determines the candidate evidence sets $D(k)$ and scales with corpus size.
Validation operates on the query-local structure induced by $D(k)$ and is independent of corpus size once retrieval is fixed.
Our results concern the correctness of this separation: whether retrieval truncation preserves the existence of feasible witnesses.

\paragraph{Noetherian guarantees versus tractability.}
Noetherian properties, such as finite witnessability and finite generation, guarantee the existence of finite certificates and uniform bounds.
They do not, by themselves, imply efficient computation of such bounds.
This distinction mirrors classical results in algebra and logic, where compactness ensures finite witnesses without guaranteeing computational efficiency.

In practice, validated retrieval exploits Noetherian structure as a correctness invariant rather than as a computational procedure.
By enforcing feasibility-preserving truncation, the system avoids failures due to missing compatible evidence without incurring combinatorial explosion.

\paragraph{Summary.}
Although general CSP feasibility is NP-hard, the feasibility preservation guarantees studied here apply to a bounded, retrieval-driven regime.
They ensure soundness of truncation without requiring solution of large or unstructured constraint systems.

\section{Discussion}

This work reframes retrieval correctness in terms of feasibility preservation rather than relevance optimization.
Under this view, failures of retrieval-based systems arise not primarily from inaccurate scoring, but from structural violations that prevent compatible evidence from co-occurring under truncation.

\paragraph{Feasibility versus relevance.}
Classical retrieval pipelines are designed to optimize marginal relevance of individual evidence items.
Our results show that such optimization is orthogonal to feasibility: highly relevant items may fail to jointly satisfy query constraints if truncation disrupts their co-occurrence.
In contrast, feasibility is an existential, global property.
Either a compatible witness exists among the retrieved evidence or it does not.

This distinction explains why retrieval systems can fail even when all necessary information is present in the corpus.
The failure lies not in relevance, but in the structure of truncation.

\paragraph{Monotonicity as a correctness invariant.}
Theorem~\ref{thm:monotone-implies-nr} identifies monotonicity as the minimal structural condition required to preserve feasibility for individual queries.
Without monotonicity, truncation may arbitrarily discard previously retrieved evidence, destroying feasible witnesses even when they exist in the limit.
Monotone retrieval ensures that once an evidence item becomes available, it remains available, enabling finite witnessability.

Importantly, monotonicity is not an algorithmic optimization but a correctness invariant.
It constrains how approximation is performed, not how evidence is ranked.

\paragraph{Uniform bounds and finite generation.}
Uniform retrieval budgets require stronger structure.
Theorem~\ref{thm:finite-gen-implies-unr} shows that finite generation of witness certificates is sufficient to guarantee a single finite depth that validates an entire class of queries.
Proposition~\ref{prop:non-finite-generation} demonstrates that this condition is also necessary.

This separation mirrors classical Noetherian phenomena: compactness ensures finite witnesses for individual instances, while finite generation is required to obtain uniform bounds across families.
In retrieval settings, this distinction clarifies why some query classes admit stable validation budgets while others inherently require unbounded expansion.

\paragraph{Global constraints and witness structure.}
Proposition~\ref{prop:slotwise-insufficient} highlights that feasibility is fundamentally relational.
Local coverage of individual slots does not guarantee global compatibility.
This motivates modeling query constraints via joint compatibility relations and validates the use of witness certificates that capture global structure rather than per-slot relevance.

In practical terms, this explains why systems that evaluate retrieval quality solely via per-slot or per-field coverage may still fail catastrophically on composed queries.

\paragraph{Implications for validated retrieval.}
Taken together, these results suggest a principled notion of validated retrieval, in which
retrieval procedures should be evaluated by whether they preserve feasibility under truncation rather than solely by relevance metrics.
Monotonicity and certificate-based reasoning provide simple, interpretable criteria for correctness, while the negative results identify sharp boundaries beyond which uniform guarantees are impossible.

\paragraph{Outlook.}
The framework developed here isolates feasibility preservation as a structural problem independent of specific retrieval architectures.
Future work may extend this analysis to richer certificate structures, alternative notions of truncation, or settings where feasibility interacts with probabilistic or approximate constraints.
More broadly, we expect feasibility-based reasoning to provide a useful lens for understanding and diagnosing failures in retrieval-augmented systems beyond the specific models considered here.

\section{Related Work}

\paragraph{Constraint satisfaction.}
The formal model used in this paper is closely related to classical constraint satisfaction problems (CSPs), where feasibility is defined as the existence of an assignment satisfying a set of constraints.
General CSP feasibility, including $k$-SAT, is NP-hard in the worst case~\cite{garey1979computers}.
Our work does not address the complexity of solving arbitrary CSPs, but instead studies structural conditions under which feasibility is preserved under truncation of variable domains.
This perspective is complementary to traditional CSP analysis, which typically focuses on algorithmic tractability rather than sound approximation under resource constraints.

\paragraph{Database query evaluation.}
The feasibility checks considered here are closely related to conjunctive query evaluation in relational databases, where query answering corresponds to checking the existence of tuples satisfying join constraints~\cite{abiteboul1995foundations}.
As in database systems, practical tractability arises from bounded query arity and aggressive pruning of candidate domains.
Our results can be viewed as identifying correctness conditions for truncation-based approximations of such existential queries.

\paragraph{Abstract interpretation and monotone approximation.}
The use of monotone truncation to preserve correctness is reminiscent of abstract interpretation and related frameworks in program analysis~\cite{cousot1977abstract}.
In these settings, soundness is ensured by over-approximating program behaviors via monotone abstractions.
Our notion of monotone retrieval plays an analogous structural role: it constrains how evidence is approximated so that feasibility properties are preserved.

\paragraph{Information retrieval.}
Classical information retrieval systems are designed to optimize relevance of individual documents or passages~\cite{manning2008introduction}.
Recent retrieval-augmented systems inherit this focus, often evaluating retrieval quality via marginal relevance metrics.
Our work highlights a complementary notion of correctness based on joint feasibility, showing that relevance optimization alone does not guarantee the preservation of compatible evidence under truncation.

\paragraph{Noetherian and compactness principles.}
The finite witnessability and finite-generation conditions identified in this paper are conceptually analogous to Noetherian and compactness principles in algebra and logic, where global properties are witnessed by finite subsets~\cite{atiyah1969introduction}.
While we do not invoke algebraic machinery directly, these analogies help clarify the distinction between existence guarantees and computational tractability.

\end{document}